\pdfoutput=1
\RequirePackage[T1]{fontenc}
\RequirePackage{fix-cm}
\RequirePackage{stix2}

\documentclass[final,number,a4paper,12pt]{elsarticle}
\biboptions{sort}

\usepackage[pass]{geometry}
\usepackage{amsmath,amsthm}
\usepackage{hyperref}
\newcommand{\ket}[1]{|#1 \rangle}

\theoremstyle{plain}%
\newtheorem{theorem}{Theorem}
\newtheorem{proposition}[theorem]{Proposition}
\newtheorem{lemma}[theorem]{Lemma}

\theoremstyle{remark}%
\newtheorem{remark}[theorem]{Remark}

\theoremstyle{definition}%

\newtheorem{definition}[theorem]{Definition}

\begin{document}
\begin{frontmatter}
\title{Information locality of a quantum locally recoverable code}
\author[1]{Ryutaroh Matsumoto\corref{cor1}}
\affiliation[1]{organization={Department of Information and Communications Engineering,
Institute of Science Tokyo},
addressline={2-12-1 Ookayama},
city={Meguro},
postcode={152-8550},
state={Tokyo},
country={Japan}}
\ead{ryutaroh@ict.e.titech.ac.jp}
\cortext[cor1]{ORCID: 0000-0002-5085-8879}
\date{3 August 2026}
\begin{abstract}
  A classical linear code $C$ of length $n$
  is said to have symbol locality $(r, \delta)$
  if for any index $j$ there exists a repair group
  $J_j \subseteq \{1, \ldots, n\}$ with $j\in J_j$
  and $|J_j| \leq r+\delta-1$
  such that any $\delta-1$ or fewer erasures in $J_j$
  can be corrected by using codeword symbols only in $J_j$.
  Later it turned out that this way of defining $r$
  overestimates the number of necessary codeword symbols for
  multiple-erasure correction, and information locality was proposed to
  define $r$ as the dimension of the punctured code of $C$ onto $J_j$.
  Recently locality $(r,\delta)$ was proposed for quantum
  error-correcting codes by following the original definition of
  symbol locality $(r, \delta)$.
  We propose a quantum counterpart of the information locality
  for quantum stabilizer codes constructed by Hermitian orthogonality,
  and a linear algebraic procedure computing a smaller repair group
  predicted by the proposed information locality and
  simultaneously reducing the number of measured observables
  in decoding to its minimum possible value.
  Then we demonstrate that the previously proposed
  definition of quantum locality $(r,\delta)$ has the
  same drawback of overestimating the number of necessary codeword symbols for
  erasure correction by providing an explicit example
  of a quantum stabilizer code.
  Finally, we will give another example of a quantum stabilizer code
  constructed by Euclidean orthogonality and two different linear codes,
  with which a natural translation of the classical information locality
  into the quantum setting underestimates
  the number of necessary codeword symbols for
  erasure correction.
\end{abstract}
\end{frontmatter}

\section{Introduction}\label{sec1}
An \textit{erasure} in quantum and classical error correction means an error
whose position in a codeword is known \citep{bennett97,grassl97,pless98}.
It is known that a quantum error-correcting code can correct twice as many erasures
as errors.
In light of this,
recent papers \citep{wu2022,kang2023} take advantage of erasures in
quantum fault-tolerant computation,
as some physical devices allow identification of qubits with erasures in a codeword
without destruction of encoded quantum information or stabilizer measurements
\citep{wu2022,kang2023}.

Measurements are costly on some physical devices
and measurement-free fault-tolerant computation has been
actively investigated recently \citep{perlin2023,heussen2024,veroni2024}.
Reducing the number of measurements in quantum error correction has also been investigated \citep{zhou25}.
In particular, measurements cause disturbance of measured qubits
on some devices \citep{perlin2023,heussen2024,veroni2024,zhou25}.
The standard procedure for quantum erasure correction 
involves measurements as described in \citep{delfosse2020}.

An approach to minimize the number of qudits acted on during an
erasure correction procedure is the quantum local recovery
\citep{golowich25}, which is based on the classical local recovery \citep{gopalan12}
of a single erasure.
Recent developments of quantum local recovery track those of
classical local recovery, some of which will be briefly reviewed here.
After \citep{gopalan12},
\citet{prakash12} considered multiple-erasure correction and
introduced $(r, \delta)$ symbol locality, which roughly allows
correction of $(\delta - 1)$ or fewer erasures by a repair group
consisting of at most $(r + \delta - 1)$ codeword symbols
including erasures. The parameter $r$ was meant to capture the
number of additional codeword symbols required for correcting
$(\delta - 1)$ erasures in its repair group.
Later the same authors \citep{kamath14} observed that
the original definition $r$ of symbol locality sometimes
overestimates the
number of required codeword symbols
and proposed $(r,\delta)$ information locality.
From the same research group,
the hierarchical local recovery was proposed \citep{sasidharan15},
in which the punctured code onto a repair group has smaller locality
$(r_2, \delta_2)$ and roughly allows correction of $(\delta_2 - 1)$
or fewer erasures by using only $r_2$ additional codeword symbols
instead of $r (> r_2)$ ones.
The classical hierarchical local recovery was defined by using the
information locality instead of the symbol locality.

After the introduction of quantum local recovery for single-erasure correction
\citep{golowich25},
many results followed \citep{aqlrc26,luo25,li2025optimalquantumlrcshermitian,li2025improvedboundsoptimalconstructions,sharma25,xie25,bu2025quantumlocallyrecoverablecode} for single-erasure correction.
Following the classical development as reviewed above,
\citet{qlrc24} proposed the quantum $(r,\delta)$ locality correcting
multiple erasures,
which was a direct translation of the classical $(r, \delta)$ symbol locality \citep{prakash12}.
We also observe very active research for quantum local recovery of multiple erasures
\citep{cao2025optimalquantumrdeltalocallyrepairable,zhou2025optimalquantumrdeltalocallyrepairable,cao2026matrix1,cao2026matrix2,galindo26qhbch,galindo2025optimalquantumlocallyrecoverable}.

Very recently, \citet{hqlrc26} introduced
quantum hierarchical local recovery, in which the quantum symbol locality
by \citet{qlrc24} was generalized,
while the corresponding classical definition \citep{sasidharan15} used the more accurate information
locality.
Our goal is to define a quantum counterpart of the classical information locality
so that we can estimate more accurately the number of required codeword symbols
for erasure correction than
the locality defined in \cite{qlrc24}.

The classical information locality is defined as the maximum dimension of
the punctured codes onto repair groups \citep[Definition 2]{kamath14}.
This definition is based on the fact that an $[n,k, \delta]$ linear code
always has a set of $k$ codeword symbols 
which is disjoint from $(\delta - 1)$ erasures and
allows reconstruction of encoded messages consisting of $k$ symbols.
This is no longer true for quantum error-correcting codes.
For example, each codeword symbol
in the binary $[[5,1,3]]_2$ stabilizer code \citep{gottesman96}
has no information, which can be seen from the fact that the density
matrix of every  symbol of any codeword
is $(1/2) I_{2 \times 2}$
independently of encoded messages.
This is a stark contrast to a binary $[n,1]_2$ linear code $C$ whose
dual has the minimum Hamming distance $\geq 2$.
Every codeword bit in $C$ varies with encoded $1$-bit messages,
which can be reconstructed from any codeword bit in $C$.
So we cannot use the number of information symbols
in the punctured \emph{quantum} error-correcting code onto a repair group
as a definition of quantum information locality,
as it underestimates the number of required codeword symbols.
We will propose another algebraic quantity that
is a tighter upper bound on the number of required codeword symbols
than the quantum symbol locality proposed in \cite{qlrc24}.

After reviewing necessary notions in Section \ref{sec2},
in order to capture the locality of quantum erasure correction
in a more accurate way,
in Section \ref{sec3}
we will propose quantum $(r,\delta)$ information
locality for the quantum stabilizer codes \citep{gottesman96,calderbank97,calderbank98,ashikhmin00,ketkar06} defined by the Hermitian dual-containing
classical linear codes.
We also propose a linear algebraic procedure computing a smaller repair group
  predicted by the proposed information locality and
  simultaneously computing observables to be measured
  in quantum local recovery
  that reduce the number of measured observables
  to its minimum possible value.
  In Section \ref{sec:diff}
  an advantage of quantum information locality over symbol one
is demonstrated by an explicit small quantum stabilizer code
with which the information locality is smaller than the symbol one,
which means the proposed locality more accurately captures the number of required
codeword symbols for erasure correction.
Finally, in Section \ref{sec4} we give an obstacle
for defining the quantum information locality
for a wider class of quantum error-correcting codes than
the class of quantum stabilizer codes considered in this paper.

\section{Preliminaries}\label{sec2}
Let $q$ be a prime power, $n$ a positive integer and $\mathbf{F}_{q^2}$
the finite field with $q^2$ elements.
We will consider quantum error-correcting codes (QECCs) of length $n$.
For $\vec{x}=(x_1, \ldots, x_n)$ and $\vec{y}=(y_1, \ldots, y_n)\in \mathbf{F}_{q^2}^n$,
their Hermitian inner product is defined by
\begin{equation*}
  \langle \vec{x}, \vec{y}\rangle_h
  = \sum_{i=1}^n x_i^q y_i,
\end{equation*}
and the Hamming weight of $\vec{x}$ is denoted  by $w_H(\vec{x})$.
For a set $C$ of vectors, by $w_H(C)$ we denote the minimum
Hamming weight
of \emph{nonzero} vectors in $C$.
For a subset $J \subset \{ 1, \ldots, n\}$ and
$\vec{x}=(x_1, \ldots, x_n)$,
$\pi_J(\vec{x})$ denotes the projected vector $(x_j)_{j \in J} \in \mathbf{F}_{q^2}^{|J|}$.
For an $\mathbf{F}_{q^2}$-linear space $C \subset \mathbf{F}_{q^2}^n$
and $J \subset \{ 1, \ldots, n\}$,
the punctured code of $C$ onto $J$
is defined as
\begin{equation*}
  \pi_J(C) = \{ \pi_J(\vec{x}) : \vec{x} \in C\},
\end{equation*}
the shortened code of $C$ onto $J$
is defined as
\begin{equation*}
  \sigma_J(C) = \{ \pi_J(\vec{x}) : \vec{x} \in C \text{ such that }
  \pi_{\overline{J}}(\vec{x}) = \vec{0} \}
\end{equation*}
as in \citep{pless98},
where $\overline{J} = \{ 1, \ldots, n\} \setminus J$,
and the dual code of $C$ with respect to the Hermitian
inner product is defined as
\begin{equation*}
  C^{\perp h} = \{ \vec{y} \in \mathbf{F}_{q^2}^n :
  \langle \vec{x}, \vec{y}\rangle_h = 0 \text{ for all } \vec{x} \in C\}.
\end{equation*}
It was noted in \citep{galindo19} that
\begin{equation}\label{eq:dual}
  \begin{split}
  \pi_J(C)^{\perp h} &= \sigma_J(C^{\perp h}),\\
  \sigma_J(C)^{\perp h} &= \pi_J(C^{\perp h}).
  \end{split}
\end{equation}

Let $\mathcal{H}_q$ be the $q$-dimensional complex linear space
$\mathbf{C}^q$, which is the state space of a single qudit.
A $q$-ary QECC of length $n$ is a complex subspace $Q \subset \mathcal{H}_q^{\otimes n}$.
Each vector $\vec{x} \in \mathbf{F}_{q^2}^n$
defines the quantum error $M(\vec{x})$ acting on $n$ qudits \citep{ashikhmin00,ketkar06}.
A QECC $Q$ is said to have distance $d$ if
there exists a vector $\vec{x}\in \mathbf{F}_{q^2}^n$ with $w_H(\vec{x}) =d$
such that $M(\vec{x})Q = Q$ and $M(\vec{x})$ changes some quantum
codeword $\ket{\varphi} \in Q$ to another quantum codeword,
and for all vectors $\vec{y} \in \mathbf{F}_{q^2}^n$ with $w_H(\vec{y}) <d$
and all quantum codewords $\ket{\varphi} \in Q$
we have either
$M(\vec{y})Q \perp Q$ or $M(\vec{y})\ket{\varphi}$ is
a scalar multiple of $\ket{\varphi}$.
The code $Q$ with distance $d$
is said to be \emph{pure} if every quantum error $M(\vec{y})$
sends $Q$ to another space orthogonal to $Q$ for $w_H(\vec{y}) < d$
and is said to be \emph{impure} otherwise.
A QECC $Q \subset \mathcal{H}_q^{\otimes n}$ of dimension
$q^k$ and distance $d$ is said to be
an $[[n,k,d]]_q$ code \citep{calderbank98},
and it can correct $t$ errors and $e$ erasures simultaneously
if $2t + e < d$.
It was shown in \citep{calderbank98,ashikhmin00,ketkar06}
that a dual-containing space $C \supsetneq C^{\perp h}$
defines an $[[n, 2 \dim C - n, w_H(C\setminus C^{\perp h})]]_q$
code $Q(C)$ that is called the stabilizer code.

Let $I \subsetneq J \subseteq \{1, \ldots, n\}$.
It was shown in \citep{qlrc24} that
erasures in $I$ can be corrected by quantum measurements
and unitary operations acting on qudits only in $J$
if and only if
\begin{equation}\label{eq:cond1}
  \sigma_I(\pi_J(C)) = \sigma_I(C^{\perp h}),
\end{equation}
which can be seen equivalent to
\begin{equation}\label{eq:cond2}
  \pi_I (C) = \pi_I (\sigma_J(C^{\perp h})),
\end{equation}
by using (\ref{eq:dual}).
\Citet{qlrc24} also showed that
any erasure correction procedure for the punctured
stabilizer code $Q(\pi_J(C))$
can be used for correcting erasures in $I$
if (\ref{eq:cond1}) or (\ref{eq:cond2}) holds.

\section{Information locality}\label{sec3}
\subsection{Definitions of locality}
Consider an $\mathbf{F}_{q^2}$-linear code $C \subset \mathbf{F}_{q^2}^n$.
The code $C$ is said to have classical symbol locality $(r_{C,s},\delta_C)$
\cite[Definition 1]{kamath14}
if for each index $j$ there exists $J_j \subseteq \{1,2, \ldots, n\}$
with $j \in J_j$
such that the punctured code $\pi_{J_j}(C)$ has the minimum Hamming
distance $\geq \delta_C$ and
\begin{equation}\label{eq:coord}
  |J_j| \leq r_{C,s} + \delta_C - 1.
\end{equation}
In this paper we will call $\delta_C$ the \emph{classical local distance}.
It was observed in \citep{kamath14,grezet19} that the number of necessary codeword
symbols for correcting $\delta_C-1$ erasures is often less than
$|J_j| - \delta_C + 1$. In order to better estimate the number of necessary
codeword symbols for correcting $\delta_C-1$ erasures,
the information locality \citep[Definition 2]{kamath14}
was defined as follows:
The code $C$ is said to have classical information locality $(r_{C,i},\delta_C)$
if the same set of conditions holds with (\ref{eq:coord}) replaced by
\begin{equation}\label{eq:Cinfo}
  \dim \pi_{J_j}(C) \leq r_{C,i}.
\end{equation}
With the information locality there is no condition on
the size of $J_j$.
The information locality for classical linear codes
was later generalized to classical quasi-uniform codes
and called dimension locality
\citep[Definition 4]{grezet19}.
Since every codeword symbol in $\pi_J(C)$ can
be reconstructed by at most $\dim \pi_J(C)$ symbols,
classical information locality captures the local property
more accurately.
The set $J_j$ is called a repair group for index $j$
in both definitions of locality.

If an $[n,k,d]$ linear code has either information or
symbol locality $(r, \delta)$, then those parameters must satisfy
the classical Singleton-like bound \citep{prakash12,kamath14}
\begin{equation}\label{eq:SingletonC}
  k+d + \left( \left\lceil \frac{k}{r} \right\rceil - 1 \right)(\delta - 1) \leq n+1.
\end{equation}

Consider a QECC $Q \subset \mathcal{H}_q^{\otimes n}$.
The code $Q$ is said to have quantum
symbol locality $(r_{Q,s},\delta_Q)$ \citep{qlrc24}
if for each index $j$ there exists $J_j \subseteq \{1,2, \ldots, n\}$
with $j \in J_j$
such that any $\delta_Q - 1$ or fewer erasures in $J_j$ satisfying (\ref{eq:coord})
can be corrected by a quantum operation
(completely positive trace-preserving map) 
acting on codeword qudits in $J_j$ of $Q$.
In this paper we call $\delta_Q$  the \emph{quantum local distance}.

\subsection{Information locality of quantum locally recoverable  codes by Hermitian orthogonality}
Fix a Hermitian dual-containing
linear code $C \subset \mathbf{F}_{q^2}^n$
and we will consider the stabilizer code $Q(C)$ defined by it.
We also fix
an erasure index $j$, its repair group $J_j \subset \{1, \ldots, n\}$ and
the  set
$I \subset J_j$ of actual erasures.
Recall that $J \subset \{1, \ldots, n\}$ can be a repair
group for the erasure set $I$ if and only if (\ref{eq:cond2}) holds.
\begin{proposition}\label{lem1}
  For a given erasure set $I \ni j$ and a repair group $J_j$
  satisfying (\ref{eq:cond2}) with $J$ replaced by $J_j$,
  there exists a subset $J \subseteq J_j$
  with $J \supseteq I$ such that $J$ satisfies (\ref{eq:cond2})
  and
  \begin{equation}\label{eq:size}
    |J| - |I| \leq \dim \pi_{J_j}(C) - \dim \sigma_I(C^{\perp h}).
  \end{equation}
  Such a set $J$ can be computed in $O(|J_j|^3)$ arithmetic
  operations in $\mathbf{F}_{q^2}$.
  The computational procedure also simultaneously provides
  an $\mathbf{F}_{q^2}$-basis defining a set of
  observables needed in quantum local recovery.
\end{proposition}
\begin{proof}
  In this proof, by abuse of notation
  all vectors are assumed to have $|J_j|$ components.
  For example, $(|J_j| - |I|)$ zeros are appended
  for vectors in $\sigma_I(C^{\perp h})$.
  Decompose $\sigma_{J_j}(C^{\perp h}) = V \oplus \sigma_I(C^{\perp h}) \oplus \sigma_{J_j\setminus I}(C^{\perp h})$ as a direct sum of three linear spaces.
  By (\ref{eq:cond2}),
  \begin{equation}\label{eq:dimequal}
  \begin{split} \dim \pi_I(C) &= \dim \pi_I(\sigma_{J_j}(C^{\perp h}))\\
    &= \dim \sigma_{J_j}(C^{\perp h}) - \dim \sigma_{J_j}(C^{\perp h}) \cap \ker(\pi_I)\\
    &= \dim \sigma_{J_j}(C^{\perp h}) - \dim \sigma_{J_j\setminus I}(C^{\perp h})\\
    &= \dim V + \dim \sigma_I(C^{\perp h}).
  \end{split}
  \end{equation}
  After finding bases for
  $\sigma_I(C^{\perp h})$, $V$ and  $\sigma_{J_j\setminus I}(C^{\perp h})$
  by linear algebra,
  we can find the following $\dim \sigma_{J_j}(C^{\perp h}) \times  |J_j|$ matrix
      \begin{equation*}
    B = \begin{pmatrix} B_I \\ B_{1, \mathrm{rem}} \\ B_2 \end{pmatrix} \in \mathbf{F}_{q^2}^{\dim \sigma_{J_j}(C^{\perp h}) \times |J_j|},
    \end{equation*}
    where:
    \begin{itemize}
    \item $B_I \in \mathbf{F}_{q^2}^{\dim \sigma_I(C^{\perp h}) \times |J_j|}$ forms a basis for $\sigma_I(C^{\perp h})$.
      The submatrix $B_I|_{J_j \setminus I}$ of $B$ consisting of
      columns corresponding to $J_j \setminus I$
      is the zero matrix.
        \item $B_{1, \mathrm{rem}} \in \mathbf{F}_{q^2}^{(\dim \pi_I (C)  - \dim \sigma_I(C^{\perp h})) \times |J_j|}$ forms a basis for  $V$.
        \item $B_2 \in \mathbf{F}_{q^2}^{(\dim \sigma_{J_j}(C^{\perp h})-\dim \pi_I (C)) \times |J_j|}$ forms a basis for $\sigma_{J_j\setminus I}(C^{\perp h})$. The submatrix $B_2|_I$ is the zero matrix.
    \end{itemize}

    Secondly, compute  a reduced row echelon form $R_2$ of $B_2$.
    Let $P \subseteq J_j \setminus I$ be the set of pivot column indices of $R_2$, with size $|P| = \mathrm{rank}(B_2) = \dim \sigma_{J_j}(C^{\perp h})-\dim \pi_I (C)$.

    Thirdly,
    eliminate pivot columns $P$ in $B_{1, \mathrm{rem}}$ using rows of $R_2$:
    \begin{equation}\label{eq:btilde}
    \widetilde{B}_{1, \mathrm{rem}} = B_{1, \mathrm{rem}} - M R_2,
    \end{equation}
    where a matrix $M$ is chosen to zero out all columns in $P$. Let $\vec{v}_1, \ldots, \vec{v}_{\dim \pi_I (C)  - \dim \sigma_I(C^{\perp h})}$ denote the rows of $\widetilde{B}_{1, \mathrm{rem}}$.

    For a vector $\vec{v}=(v_1, \ldots, v_{|J_j|})$
    by $\mathrm{supp}(\vec{v})$ we denote
    $\{ \ell : v_\ell \neq 0 \}$.
    Let
    \begin{equation*}
    S = \bigcup_{m=1}^{\dim \pi_I (C)  - \dim \sigma_I(C^{\perp h})} \left( \mathrm{supp}(\vec{v}_m) \setminus I \right) \subseteq (J_j \setminus I) \setminus P,
    \end{equation*}
    and we obtain $J = I \cup S \subseteq J_j$.
    The overall computational complexity is easily seen as
    $O(|J_j|^3)$.

    We will verify $\pi_I(C) = \pi_I(\sigma_J(C^{\perp h}))$ of condition (\ref{eq:cond2}).
    For a matrix $A$, by $\langle A \rangle$ we denote its row space.
    Because $\pi_I\langle B_2 \rangle = \{ \vec{0} \}$,
    subtraction of $M R_2$ in (\ref{eq:btilde})
    leaves columns of $B_{1, \mathrm{rem}}$ indexed by $I$ unchanged:
\begin{equation*}
\pi_I(\langle \widetilde{B}_{1, \mathrm{rem}} \rangle ) = \pi_I(\langle B_{1, \mathrm{rem}}\rangle ).
\end{equation*}
Thus, the rows of $B_I$ together with $\widetilde{B}_{1, \mathrm{rem}}$
generate a subspace of $\sigma_J(C^{\perp h})$ whose projection onto $I$ spans $\pi_I(\sigma_{J_j}(C^{\perp h})) = \pi_I(C)$. Every row vector $\vec{u}$ in $B_I$ or $\widetilde{B}_{1, \mathrm{rem}}$ belongs to
$\sigma_{J_j}(C^{\perp h})$ and
has support contained in $J = I \cup S$.
Consequently, $\pi_J(\vec{u}) \in \sigma_J(C^{\perp h})$,
establishing $\pi_I(C) \subseteq \pi_I(\sigma_J(C^{\perp h}))$.
The reverse inclusion $\pi_I(C) \supseteq \pi_I(\sigma_J(C^{\perp h}))$
follows from $C \supsetneq C^{\perp h}$. Therefore
(\ref{eq:cond2}) holds.

Next, we establish the size bound (\ref{eq:size}).
Since all pivot columns $P$ are zeroed out in $\widetilde{B}_{1, \mathrm{rem}}$, the support set $S$ is restricted to $(J_j \setminus I) \setminus P$. Counting symbols gives:
\begin{align}
  |S| \le |J_j \setminus I| - |P| & = (|J_j| - |I|) - (\dim \sigma_{J_j}(C^{\perp h}) - \dim \pi_I(C)) \notag \\
  & = (|J_j| - \dim \sigma_{J_j}(C^{\perp h})) - (|I| - \dim \pi_I(C)). \label{eq:s_bound_step}
\end{align}
By (\ref{eq:dual}) we see $|I| - \dim \pi_I(C) = \dim \sigma_I (C^{\perp h})$.
Substituting this into (\ref{eq:s_bound_step}) and noting
that $\dim \pi_{J_j}(C) = |J_j| - \dim \sigma_{J_j}(C^{\perp h})$ yields:
\begin{equation*}
  |J| - |I| = |S| \leq (|J_j| - \dim \sigma_{J_j}(C^{\perp h})) - \dim \sigma_I (C^{\perp h}) = \dim \pi_{J_j}(C) - \dim \sigma_I(C^{\perp h}).
\end{equation*}

Finally, we will clarify that the computational procedure in this proof
provides a set of observables measured in quantum local recovery.
Recall that quantum local recovery with a repair group $J_j$
is standard erasure correction on punctured quantum codewords
in $Q(\pi_J(C))$, whose stabilizer is defined by $\sigma_{J_j}(C^{\perp h})$
\citep{qlrc24}.
Observables defined by basis vectors for $\sigma_{J_j}(C^{\perp h})$
can be used for measurement in the above local recovery procedure.
On the other hand, since erasures exist only in $I$,
observables corresponding to
$\sigma_{J_j\setminus I}(\sigma_{J_j}(C^{\perp h})) = \sigma_{J_j\setminus I}(C^{\perp h})$
are of no use for identifying erasures in $I$, which was formally proved in \citep{qlrsurface25eprint}.
From the observations in this proof we see that the computed $J$
satisfies
\begin{equation}\label{eq:decomposition}
  \sigma_{J_j}(C^{\perp h}) =   \sigma_{J_j\setminus I}(C^{\perp h}) \oplus \sigma_J(C^{\perp h}), 
\end{equation}
and that the row vectors of $B_I$ and $\widetilde{B}_{1, \mathrm{rem}}$
form a $\mathbf{F}_{q^2}$-basis for $\sigma_J(C^{\perp h})$.
As condition (\ref{eq:cond2}) holds for the computed $J$,
we can just measure observables defined by row vectors of
$B_I$ and $\widetilde{B}_{1, \mathrm{rem}}$
and identify the erasures in $I$, which completes the proof
of the last sentence in Proposition \ref{lem1}.
\end{proof}

\begin{remark}
  As mentioned in Section \ref{sec1},
  since measurements often disturb quantum states of measured qudits,
  it is also important to reduce the number of measured observables
  as well as that of measured qudits.
  By \citep[Theorem 5]{qlrsurface25eprint} and (\ref{eq:decomposition})
  we see that the observables computed in the proof of Proposition \ref{lem1}
  attain the minimum possible number of observables for
  the punctured code $Q(\pi_{J_j}(C))$ to correct erasures in $I$.
\end{remark}

\begin{definition}\label{def:qil}
In the same spirit as the classical information locality,
the code $Q(C)$ is said to have quantum
information locality $(r_{Q,i},\delta_Q)$, which is a novel notion in this paper,
if for each index $j$ there exists $J_j \subseteq \{1,2, \ldots, n\}$
with $j \in J_j$ such that 
we have $w_H(\pi_{J_j}(C) \setminus \sigma_{J_j}(C^{\perp h})) \geq \delta_Q$
and $\dim \pi_{J_j}(C) \leq r_{Q,i}$.
\end{definition}
We will omit ``classical'' and ``quantum'' from locality if
it is clear from context.
Observe that $r_{Q,i}$ is defined in exactly the same way as $r_{C,i}$.
Observe also that for a fixed $C$ we have $r_{C,i} = r_{Q,i}$ but
$r_{C,s} \neq r_{Q,s}$ if $\delta_C \neq \delta_Q$.

The operational meaning of the quantum information locality is clarified below.
\begin{proposition}
  Suppose that $Q(C)$ has the information locality
  $(r_{Q,i},\delta_Q)$ and there is an erasure at $j$-th codeword qudit.
  Then erasures in $I (\subsetneq J_j)$ can be corrected by using
  at most $r_{Q,i}$ additional codeword qudits if $|I| \leq \delta_Q - 1$.
\end{proposition}
\begin{proof}
As shown in \citep{qlrc24},
the local erasure correction on $J_j$ can be
done by a standard erasure correction procedure
for $Q(\pi_{J_j}(C)) \subset \mathcal{H}_q^{\otimes |J_j|}$.
Therefore, since the Hermitian dual of $\pi_{J_j}(C)$
is $\sigma_{J_j}(C^{\perp h})$ by (\ref{eq:dual}),
the maximum number of correctable erasures is
$w_H(\pi_{J_j}(C) \setminus \sigma_{J_j}(C^{\perp h})) - 1$.
In addition to the codeword qudits erased,
by Proposition \ref{lem1}, the number of additional codeword
qudits necessary for erasure correction is $\leq \dim \pi_{J_j}(C)$.
\end{proof}

In \citep{qlrc24} relations between
the classical and the quantum localities were clarified.
We will clarify relations between information localities below.
They can be proved in the same way as \citep{qlrc24}
so their proofs will be omitted.

\begin{proposition}\label{prop:cmeansq}
  Let $C \subset \mathbf{F}_{q^2}^n$
  with $C \supsetneq C^{\perp h}$.
  Assume that  $C$ has classical information locality
  $(r_{C,i}, \delta_C)$.
  Then $Q(C)$ has quantum information locality
  $(r_{C,i}, \delta_C)$. \qed
\end{proposition}

By Proposition \ref{prop:cmeansq},
one can construct a QECC $Q(C)$ with quantum information locality
$(r, \delta)$ by designing a classical code $C$ with
classical information locality $(r, \delta)$ such that
$C \supsetneq C^{\perp h}$.
A converse of Proposition \ref{prop:cmeansq}
needs an additional assumption.
\begin{proposition}\label{prop:qmeansc}
  Let $C \subset \mathbf{F}_{q^2}^n$
  with $C \supsetneq C^{\perp h}$.
  Assume that  $Q(C)$ has quantum information locality
  $(r_{Q,i}, \delta_Q)$.
  If
  \begin{enumerate}
  \item either $\delta_Q \leq w_H(C^{\perp h})$ or
  \item $Q(C)$ is pure, that is, $w_H(C) = w_H(C \setminus C^{\perp h})$
  \end{enumerate}
  then $C$ has classical information locality
  $(r_{Q,i}, \delta_Q)$. \qed
\end{proposition}

By (\ref{eq:SingletonC}) and Proposition \ref{prop:qmeansc},
if an $[[n,k,d]]$ \emph{pure} stabilizer code $Q(C)$ has either information or
symbol locality $(r,\delta)$ then those parameters must satisfy
\begin{equation}\label{eq:SingletonQ}
  \frac{k+n}{2} +d + \left( \left\lceil \frac{k+n}{2r} \right\rceil - 1 \right)(\delta - 1) \leq n+1
\end{equation}
which was called the quantum \emph{pure} Singleton-like bound
for quantum local recovery in \citep{qlrc24}.
\citet{impure26} showed that (\ref{eq:SingletonQ}) can be violated
if $Q(C)$ is impure.

\section{Example differentiating information and symbol localities}\label{sec:diff}
\begin{proposition}\label{prop:diff}
  There exist $\mathbf{F}_4$-linear codes
  $C^{\perp h} \subsetneq C \subsetneq \mathbf{F}_4^{16}$
  whose stabilizer code $Q(C)$ has parameters $[[16, 2, 3]]_2$
  with symbol locality $(6,3)$ and information locality $(5,3)$.
\end{proposition}
\begin{proof}
  We construct a binary quantum stabilizer code $Q(C)$
  from a classical linear code over
  the finite field $\mathbf{F}_4 = \{0, 1, \omega, \omega^2\}$ (where $\omega^2 + \omega + 1 = 0$) of length $n=16$. 

We begin by defining a $3 \times 8$ matrix $H_0$ whose columns are pairwise linearly independent over $\mathbf{F}_4$:
\begin{equation*}
H_0 = \begin{pmatrix}
1 & 1 & 1 & 1 & 1 & 1 & 1 & 1 \\
0 & 1 & \omega & \omega^2 & 0 & 1 & \omega & \omega^2 \\
0 & 0 & 0 & 0 & 1 & 1 & 1 & 1
\end{pmatrix}.
\end{equation*}
By evaluating the Hermitian inner products of the rows of $H_0$,
we verify that $H_0 H_0^\dagger = \mathbf{0}$,
where $H_0^\dagger$ is the transpose of $H_0$ followed
by component-wise squaring.
We identify a vector $\vec{v} \in \langle H_0\rangle^{\perp h} \setminus \langle H_0 \rangle$ defined as $\vec{v} = (0, 1, \omega^2, \omega, 0, 0, 0, 0)$.
Its Hermitian inner product with itself is $\langle \vec{v}, \vec{v} \rangle_h = 1$.

We construct the generator matrix $G^{\perp h}$ of the dual code $C^{\perp h}$ for $n=16$ by placing two identical blocks of $H_0$ on the diagonal and appending a global parity-check row $(\vec{v}, \vec{v})$ to glue them together:
\begin{equation*}
G^{\perp h} = \begin{pmatrix}
H_0 & \mathbf{0} \\
\mathbf{0} & H_0 \\
\vec{v} & \vec{v}
\end{pmatrix}.
\end{equation*}
The matrix $G^{\perp h}$ has size $7 \times 16$. We
see the self-orthogonality condition $G^{\perp h} (G^{\perp h})^\dagger = \mathbf{0}$
from
\begin{itemize}
    \item The blocks $H_0$ satisfy $H_0 H_0^\dagger = \mathbf{0}$.
    \item The row $(\vec{v}, \vec{v})$ is Hermitian-orthogonal to the $H_0$ blocks because $\vec{v} \in \langle H_0 \rangle^{\perp h}$.
    \item The Hermitian inner product of the row $(\vec{v}, \vec{v})$ with itself is $\langle \vec{v}, \vec{v} \rangle_h + \langle \vec{v}, \vec{v} \rangle_h = 0$.
\end{itemize}
The dimension of the dual code is $\dim C^{\perp h} = 7$.

The primary code $C = \langle G \rangle$ has dimension $\dim C = n - \dim C^{\perp h} = 16 - 7 = 9$.
To explicitly construct the generator matrix $G$ for $C$, we find a weight-3 vector $\vec{u} \in \langle H_0 \rangle^{\perp h}$ with a support completely disjoint from $\vec{v}$. We define $\vec{u} = (0, 0, 0, 0, 0, 1, \omega^2, \omega)$. Because their supports are disjoint, $\langle \vec{u}, \vec{v} \rangle_h = 0$. We form the remaining basis vectors $G_{\mathrm{rem}}$ using $\vec{u}$:
\begin{equation*}
G_{\mathrm{rem}} = \begin{pmatrix}
\vec{u} & \vec{0} \\
\vec{0} & \vec{u}
\end{pmatrix} = \begin{pmatrix}
0 & 0 & 0 & 0 & 0 & 1 & \omega^2 & \omega & 0 & 0 & 0 & 0 & 0 & 0 & 0 & 0 \\
0 & 0 & 0 & 0 & 0 & 0 & 0 & 0 & 0 & 0 & 0 & 0 & 0 & 1 & \omega^2 & \omega
\end{pmatrix},
\end{equation*}
\begin{equation*}
G = \begin{pmatrix} G^{\perp h} \\ G_{\mathrm{rem}} \end{pmatrix}.
\end{equation*}
This explicitly demonstrates that $G^{\perp h}$ is a submatrix of $G$, satisfying $C \supsetneq C^{\perp h}$. We proved that $Q(C)$ has parameters $[[16,2]]$.

We will compute $w_H(C^{\perp h})$.
Any non-zero linear combination of rows strictly within an $H_0$ block yields a weight $\ge 4$. If a vector includes the row $(\vec{v}, \vec{v})$,
its nonzero scalar multiple takes the form $\vec{y} = (\vec{x}_1 + \vec{v}, \vec{x}_2 + \vec{v})$, where $\vec{x}_1, \vec{x}_2 \in  \langle H_0 \rangle$.
Because $\vec{v} \notin \langle H_0 \rangle$, neither $\vec{x}_1+\vec{v}$ nor $\vec{x}_2+\vec{v}$ can be the zero vector $\vec{0}$. Because $H_0$ can be seen as
parity-check for $\vec{x}_1+\vec{v}$ and $\vec{x}_2+\vec{v}$
and $H_0$ contains no parallel columns,
we observe
$w_H(\vec{x}_1+\vec{v}) \ge 3$ and $w_H(\vec{x}_2+\vec{v}) \ge 3$, meaning $w_H(\vec{y}) \ge 6$. The absolute minimum weight is bounded solely by the $H_0$ subcodes
and we see
\begin{equation*}
w_H(C^{\perp h}) = 4.
\end{equation*}

We will compute  $w_H(C)$ and $w_H(C \setminus C^{\perp h})$.
The vector $\vec{c} = (\vec{u}, \vec{0})$
corresponds to the first row of $G_{\mathrm{rem}}$ and
belongs to $C$.
Counting its non-zero elements gives an immediate upper bound
\begin{equation*}
w_H(\vec{c}) = w_H(\vec{u}) = 3 \implies w_H(C) \le 3.
\end{equation*}
Because $G^{\perp h}$ contains no parallel or zero columns, no weight-1 or weight-2
codewords exist in $C$. Therefore, $w_H(C) = 3$.
Since $w_H(C^{\perp h})=4$, the distance of $Q(C)$ is
\begin{equation*}
w_H(C \setminus C^{\perp h}) = 3.
\end{equation*}

We will define repair groups and compute their quantum local distances.
Let $K = \{1, \ldots, 8\}$ and $L = \{9, \ldots, 16\}$,
which will be repair groups for $C$.
For erased position $1 \leq j \leq 8$, the repair group $K$ is used
and otherwise $L$ is used.

The shortened dual code $\sigma_K(C^{\perp h})$ restricts
vectors to be zero on $L$.
This forces the coefficient of the row $(\vec{v}, \vec{v})$
to be exactly $0$, isolating the top-left block
\begin{align*}
\sigma_K(C^{\perp h}) &= \langle H_0 \rangle \implies \dim \sigma_K(C^{\perp h}) = 3, \\
\pi_K(C) &= \langle H_0 \rangle^{\perp h} \implies \dim \pi_K(C) = 8 - 3 = 5.
\end{align*}
Because the vector $\vec{u}$ is completely captured in this projection,
we see the minimum Hamming weight
\begin{equation*}
w_H(\pi_K(C)) = w_H(\vec{u}) = 3.
\end{equation*}
Because $\sigma_K(C^{\perp h}) = \langle H_0 \rangle$, its minimum Hamming
weight is $4$.
Therefore, we see
\begin{equation*}
w_H(\pi_K(C) \setminus \sigma_K(C^{\perp h})) = 3.
\end{equation*}
We have the same analysis for the other repair group $L$,
and we see that the quantum local distance $\delta_Q =3$.
We also see that information locality $(r_{Q,i}, \delta_Q) = (5,3)$
and symbol locality $(r_{Q,s}, \delta_Q) = (6,3)$.
\end{proof}

In order to differentiate the symbol and the information localities,
we must ensure that there is no choice of repair groups
that makes $(r_{Q,s}, \delta_Q) = (5,3)$.
For this purpose, we will introduce notations and lemmas.
For $S \subseteq \{1, \ldots, 8\}$,
$H_0(S)$ denotes the submatrix of $H_0$ consisting
of the $i$-th columns for $i \in S$.
For $i$, $\vec{c}_i$ denotes the $i$-th column of $H_0$.
\begin{lemma}\label{lem2}
  For a subset $S \subset \{1, \ldots, 8\}$
  with $|S| \ge 2$,  we have $\mathrm{rank}(H_0(S)) \ge 2$. \qed
\end{lemma}

\begin{lemma}\label{lem3}
  For a subset  $S \subset \{1, \ldots, 8\}$
  with $|S| \ge 5$,  we have $\mathrm{rank}(H_0(S)) =3$.
\end{lemma}
\begin{proof}
Let $V_1 = \mathrm{span}(\vec{c}_1, \ldots, \vec{c}_4)$ and
$V_2 = \mathrm{span}(\vec{c}_5, \ldots, \vec{c}_8)$.
Since every vector in $V_1$ has a third coordinate of $0$,
$V_1 = \{ (a, b, 0)^T \mid a, b \in \mathbf{F}_4 \}$.
Conversely, every column in the second half of $H_0$
takes the form $(1, x, 1)^T$. Because its third coordinate is $1 \neq 0$, no column from $\vec{c}_5, \ldots, \vec{c}_8$ exists in $V_1$.
We also observe $V_2 = \{ (a, b, a)^T \mid a, b \in \mathbf{F}_4 \}$
and no column from $\vec{c}_1, \ldots, \vec{c}_4$ exists in $V_2$.

Let $S = S_1 \cup S_2$ where
$S_1 = S \cap \{1, 2, 3, 4\}$ and $S_2 = S \cap \{5, 6, 7, 8\}$.
Let $s_1 = |S_1|$ and $s_2 = |S_2|$. We are given $s_1 + s_2 = |S| \ge 5$,
which implies $1 \leq s_1 \leq 4$ and $1 \leq s_2 \leq 4$.
Furthermore, $s_1 + s_2 \ge 5$ implies
that either $s_1 \ge 2$ or $s_2 \ge 2$.

Without loss of generality, assume $s_1 \ge 2$. By Lemma \ref{lem2},
the columns in $S_1$ span the 2-dimensional space $V_1$.
Because $s_2 \ge 1$, there exists at least one index $i \in S_2$. As established,
$\vec{c}_i \notin V_1$.
Adding this linearly independent vector $\vec{c}_i$ to
the 2-dimensional subspace $V_1$ strictly increases
its dimension to 3, proving $\mathrm{rank}(H_0(S)) = 3$.
\end{proof}

\begin{proposition}\label{prop:diff2}
  For the quantum stabilizer code $Q(C)$ constructed in
  the proof of Proposition \ref{prop:diff},
  no choice of repair groups makes the symbol locality
  $(r_{Q,s}, \delta_Q) = (5,3)$.
\end{proposition}
\begin{proof}
  In order to have
  symbol locality
  $(r_{Q,s}, \delta_Q) = (5,3)$,
  there must exist a repair group $J \subsetneq \{1, \ldots, 16\}$
  such that
  \begin{itemize}
  \item $|J| = 7$, and
  \item $w_H(\pi_J(C) \setminus \sigma_J(C^{\perp h})) \geq 3$.
  \end{itemize}
  In order to prove Proposition \ref{prop:diff2},
  we will show that
  for any $J$ with $|J|=7$,
  we have $w_H(\pi_J(C) \setminus \sigma_J(C^{\perp h})) \leq 2$.

  Since $\sigma_J$ can be seen as the kernel of linear map
  $\pi_{\overline{J}}$, we see that
  \begin{equation}\label{eq:dim2}
    \dim \sigma_J(C^{\perp h}) = \dim C^{\perp h} - \dim \pi_{\overline{J}} (C^{\perp h}).
  \end{equation}
  Firstly we consider the case $J \subsetneq \{1, \ldots, 8\}$.
  We have $\sigma_J(C^{\perp h}) = \sigma_J(\langle H_0 \rangle)$.
  Since $w_H(\langle H_0 \rangle) = 4$ as seen in the proof of Proposition \ref{prop:diff},
  by (\ref{eq:dim2}) we have
  \begin{equation*}
    \dim \sigma_J(\langle H_0 \rangle)  = \dim \langle H_0 \rangle - \underbrace{\dim \pi_{\{1, \ldots, 8\} \setminus J}(\langle H_0 \rangle)}_{= 1}  = 2.
    \end{equation*}
  For the case $J \subsetneq \{9, \ldots, 16\}$
  we also see that $\dim \sigma_J(C^{\perp h})= 2$.

  We consider the case $J$ is contained in
  neither $\{1, \ldots, 8\}$ nor $\{9, \ldots, 16\}$.
  Let $T_1 = \{1, \ldots, 8 \} \setminus J$ and
  $T_2 = \{ 1 \leq i \leq 8  : 8 + i  \notin J \}$.
  We have $J \cup T_1 \cup (8 + T_2) = \{1, \ldots, 16 \}$,
  $|J| + |T_1| + |T_2| = 16$ and $|T_1| + |T_2| =9$.
  Let $t_1 = |T_1|$ and $t_2 = |T_2|$.
  Since $J$ is contained in
  neither $\{1, \ldots, 8\}$ nor $\{9, \ldots, 16\}$,
  we have $(t_1, t_2) \in  \{(2,7), (3,6), (4,5), (5,4), (6,3), (7,2)\}$,
  which implies $\max\{ t_1, t_2 \}\geq 5$.
  By using Lemma \ref{lem2}
  combined with $\min\{t_1, t_2\} \geq 2$ and
  Lemma \ref{lem3} combined with $\max\{t_1, t_2\} \geq 5$,
  we see
  \begin{equation}\label{eq100}
    \dim \pi_{T_1} (\langle H_0 \rangle) + \dim \pi_{T_2} (\langle H_0 \rangle) \geq 2 + 3 = 5.
  \end{equation}
  By the shape of $G^{\perp h}$ we also see
  \begin{equation}\label{eq101}
    \dim \pi_{\{1, \ldots, 16\} \setminus J} (C^{\perp h})
    \geq \dim \pi_{T_1} (\langle H_0 \rangle) + \dim \pi_{T_2} (\langle H_0 \rangle).
  \end{equation}
  By combining (\ref{eq:dim2}), (\ref{eq100}) and (\ref{eq101})
  we see that $\dim \sigma_J(C^{\perp h}) \leq 2$.

  When $\dim \sigma_J(C^{\perp h}) \leq 1$ it is clear
  that $w_H(\pi_J (C)) \leq 2$.
  Suppose that $\dim \sigma_J(C^{\perp h}) = 2$.
  If every pair of distinct vectors in a  subset $U \subseteq \mathbf{F}_4^2$
  is linearly independent, then $|U| \leq 5$. On the other hand,
  the size of a parity-check matrix for $\pi_J (C)$ is $2 \times 7$,
  so it must have a pair of columns that is linearly dependent.
  This means that $w_H(\pi_J(C)) \leq 2$.
  Since $w_H(C^{\perp h}) = 4$, we have $w_H(\sigma_J(C^{\perp h})) \geq 4$
  and $w_H(\pi_J(C) \setminus \sigma_J(C^{\perp h})) \leq 2$.
\end{proof}

\section{Quantum locally recoverable  codes by Euclidean orthogonality}\label{sec4}
In this section we will argue that
there is difficulty in defining a quantum information locality
for quantum stabilizer codes constructed by
Euclidean orthogonality.
Let $C_X$, $C_Z \subsetneq \mathbf{F}_q^n$
be two $\mathbf{F}_q$-linear codes.
The Euclidean dual of an $\mathbf{F}_q$-linear code
will be denoted by $^{\perp e}$.
If $C_X \supseteq C_Z^{\perp e}$ and $C_Z \supseteq C_X^{\perp e}$
then we can construct an $[[n, \dim C_X + \dim C_Z - n, \min\{ w_H(C_X \setminus C_Z^{\perp e}), w_H(C_Z \setminus C_X^{\perp e})\}]]$
stabilizer code $Q(C_X, C_Z)$ \citep{calderbank96,steane96,ashikhmin00,ketkar06}.
Let $\{ \omega, \omega^q\}$ be a normal basis of
  $\mathbf{F}_{q^2}$ over $\mathbf{F}_q$.
  $Q(C_X, C_Z)$ can be written as $Q(C)$ by $C^{\perp h} \subsetneq C \subseteq \mathbf{F}_{q^2}^n$ if and only if $C_X = C_Z$ and $C = \omega C_X + \omega^q C_Z$.
Therefore all the mathematical claims in Section \ref{sec3}
  hold for
  the Calderbank-Shor-Steane quantum codes
  constructed from $C^{\perp e} \subsetneq C \subseteq \mathbf{F}_q^n$
  if $\perp h$ is replaced by $\perp e$.

  Suppose that erasures in $I$ are corrected by a repair group $J \supsetneq I$.
  $Q(C_X, C_Z)$ can correct those erasures if and only if \citep{qlrc24}
  \begin{equation}\label{eq:cond3}
    \begin{split}
  \pi_I (C_X) &= \pi_I (\sigma_J(C_Z^{\perp e})),\\
  \pi_I (C_Z) &= \pi_I (\sigma_J(C_X^{\perp e})).
    \end{split}
  \end{equation}
  
  By using Proposition \ref{lem1} and the relations
  $C = \omega C_X + \omega^q C_Z \subseteq \mathbf{F}_{q^2}^n$,
  we have $\dim_{\mathbf{F}_{q^2}} C = (\dim_{\mathbf{F}_q} C_X + \dim_{\mathbf{F}_q} C_Z)/2$ for $C_X=C_Z$,
  and we immediately deduce the following proposition:
  \begin{proposition}\label{prop:CSS}
    Assume that $C_X = C_Z$.
  For a given erasure set $I \ni j$ and a repair group $J_j$
  satisfying (\ref{eq:cond3}) with $J$ replaced by $J_j$,
  there exists a subset $J \subseteq J_j$ with $J \supseteq I$
  such that $J$ satisfies (\ref{eq:cond3})
  and
  \begin{equation}\label{eq:sizeCSS}
    |J| - |I| \leq \left\lceil \frac{\dim_{\mathbf{F}_q} \pi_{J_j}(C_X) + \dim_{\mathbf{F}_q} \pi_{J_j}(C_Z)}{2} \right\rceil.
  \end{equation}
  \qed
\end{proposition}
  When $C_X = C_Z$,
  by Proposition \ref{prop:CSS} we can define the quantum information locality
  by
  \begin{equation}\label{eq:localCSS}
    \max_{1 \leq j \leq n} \left\lceil \frac{\dim_{\mathbf{F}_q} \pi_{J_j}(C_X) + \dim_{\mathbf{F}_q} \pi_{J_j}(C_Z)}{2} \right\rceil,
    \end{equation}
  which is equal to previously defined $r_{Q,i}$ for
  $C_X = C_Z$ and $C = \omega C_X + \omega^q C_Z$.
  
  However, when $C_X \neq C_Z$,
  $C = \omega C_X + \omega^q C_Z$ may not be $\mathbf{F}_{q^2}$-linear,
  Proposition \ref{prop:CSS} can fail
  and we cannot define a quantum information locality
  by (\ref{eq:localCSS}), which will be shown by an example below.

  We will give index sets $I \subsetneq J_j$, and codes
  $C_X \times C_Z \supsetneq C_Z^{\perp e} \times C_X^{\perp e}$
  satisfying (\ref{eq:cond3}) with $J$ replaced by $J_j$
  such that no subset $J \subseteq J_j$ with $I \subsetneq J$ satisfies both (\ref{eq:cond3}) and (\ref{eq:sizeCSS}).

  Let $q = 2$, $n = 6$, $J_j = \{1, 2, 3, 4, 5, 6\}$ for all $j$,
  and $I = \{1\}$. The cardinalities are $|J_j| = 6$ and $|I| = 1$, giving $|J_j| - |I| = 5$.
Define the component codes $C_X, C_Z \subset \mathbf{F}_2^6$ by specifying $C_Z^{\perp e}$ and $C_X^{\perp e}$:
\begin{align*}
C_Z^{\perp e} &= \mathrm{span}_{\mathbf{F}_2} \{ (1, 1, 1, 1, 1, 1) \}, \\
C_X^{\perp e} &= \mathrm{span}_{\mathbf{F}_2} \{ (1, 1, 0, 0, 0, 0), \; (0, 1, 1, 0, 0, 0), \; (0, 0, 1, 1, 0, 0) \}.
\end{align*}
Taking Euclidean duals in $\mathbf{F}_2^6$, the primal component codes are:
\begin{align*}
C_Z &= \left\{ \vec{y} \in \mathbf{F}_2^6 : \sum_{m=1}^6 y_m = 0 \right\}, \\
C_X &= \left\{ \vec{y} \in \mathbf{F}_2^6 : y_1 + y_2 = 0, \, y_2 + y_3 = 0, \, y_3 + y_4 = 0 \right\}. 
\end{align*}

By straightforward computation
we can confirm
$C_X \supsetneq C_Z^{\perp e}$, $C_Z \supsetneq C_X^{\perp e}$,
and condition (\ref{eq:cond3}) holds with $J$ replaced by $J_j$.
We will verify non-existence of valid subsets $J \subseteq J_j$
by evaluating all candidate subsets $J \subseteq J_j$ containing $I = \{1\}$:
\begin{description}
    \item[Full set candidate ($J = J_j = \{1, 2, 3, 4, 5, 6\}$):]
      The projection equalities (\ref{eq:cond3}) hold.
      However, the size difference is $|J_j| - |I| = 6 - 1 = 5$. Since $5 > 4$, the required bound in (\ref{eq:sizeCSS}) fails.
    \item[Proper subset candidates ($J \subsetneq J_j$):]
    If $J$ is a proper subset of $J_j$, the complement $\overline{J} = J_j \setminus J$ is non-empty. Any non-zero vector in $C_Z^{\perp e}$ is a scalar multiple of $\vec{1} = (1, 1, 1, 1, 1, 1)$, which has non-zero entries on all 6 coordinates. Because $\overline{J}$ contains at least one coordinate, restricting to zero on $\overline{J}$ forces the shortened code $C_Z^{\perp e}$ to be trivial: $\sigma_J(C_Z^{\perp e}) = \{\vec{0}\}$.
    Consequently, $\pi_I(\sigma_J(C_Z^{\perp e})) = \{0\} \neq \mathbf{F}_2 = \pi_I(C_X)$.
    The projection equalities (\ref{eq:cond3}) fail for all proper subsets.
\end{description}
Therefore, no subset $J \subseteq J_j$ with $I \subsetneq J$ satisfies both (\ref{eq:cond3}) and (\ref{eq:sizeCSS}), and we have confirmed
that Proposition \ref{prop:CSS} fails in this situation.

\section{Concluding remarks}\label{sec:final}
In this paper, we proposed the quantum information locality
in Definition \ref{def:qil} 
by following an idea behind its classical counterpart \citep[Definition 2]{kamath14},
and a computational procedure
in Proposition \ref{lem1}
that provides a repair group whose size is predicted by
the proposed information locality and a smallest set of observables measured
during quantum local recovery.
After that, we showed in Section \ref{sec:diff}
an example of a quantum stabilizer code
with which the proposed information locality gave a tighter
estimate on the number of required codeword symbols for
erasure correction than the previously proposed quantum symbol locality \citep{qlrc24}.
In Section \ref{sec4} we argued that it seems unobvious to
define the quantum information locality for a wider class of QECC
by showing an explicit example.
In the classical local recovery,
the $(r,\delta)$ symbol locality can be defined for any block error-correcting
codes including nonlinear ones \citep{prakash12},
and the $(r,\delta)$ information locality was only defined for
the quasi-uniform codes \citep[Definition 4]{grezet19},
which include linear codes
while they form  a proper subclass of general block error-correcting codes.
So the limited applicability of information locality seems
somewhat natural even for the quantum local recovery.

\paragraph{Acknowledgments}
The author would like to thank Carlos Galindo and Fernando Hernando
for drawing his attention to classical and quantum
hierarchical local recovery \citep{sasidharan15,hqlrc26},
which initiated this research.
This work was partially funded by the Japan Society for Promotion of Science
under Grant No.\ 23K10980.

\paragraph{Data availability}
A C program to verify all the mathematical claims in Section \ref{sec:diff}
by examining all related vectors was included in source files
of the arxiv.org eprint (version 1)
with the same title as this manuscript.

\paragraph{Declarations}
The author has no competing interests to declare that are relevant to the content of this paper.

\paragraph{Declaration of AI use}
Examples in Sections \ref{sec:diff} and \ref{sec4}
were discovered by Google Gemini 3.1Pro.
Their descriptions were written by the author.
Google Gemini 3.1Pro was also used to generate
a C program to verify all the mathematical claims in Section \ref{sec:diff}
by brute-force computation.
  

\end{document}